%% file: 0_main.tex
\documentclass{llncs}

\usepackage{graphicx} 
\usepackage{macros}
\usepackage{hyperref}
\usepackage{comment}
\usepackage{multirow}
\usepackage{subcaption}
\usepackage{adjustbox}
\usepackage{orcidlink}

\usepackage[firstpage]{draftwatermark} 

\usepackage[ruled,vlined,linesnumbered]{algorithm2e}

\SetKwInOut{Requires}{Requires}
\SetKwInOut{Ensures}{Ensures}

\SetCommentSty{MyCommentStyle}

\definecolor{mycolor}{RGB}{153, 0, 0}

\SetWatermarkAngle{0}
\SetWatermarkText{\raisebox{15.5cm}{
\hspace{0.1cm}
\href{https://doi.org/10.5281/zenodo.15221148}{\includegraphics{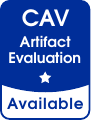}}
\hspace{8.2cm}
\includegraphics{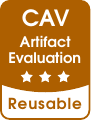}
}}

\title{Property Directed Reachability with Extended Resolution}
\author{Andrew Luka\orcidlink{0009-0007-9369-859X} \and Yakir Vizel\orcidlink{0000-0002-5655-1667}}
\institute{Computer Science Department, Technion, Haifa, Israel}

\begin{document}

\maketitle


\begin{abstract}

Property Directed Reachability (\Pdr), also known as IC3, is a state-of-the-art model checking algorithm widely used for verifying safety properties. While \Pdr is effective in finding inductive invariants, its underlying proof system, Resolution, limits its ability to construct short proofs for certain verification problems. 

This paper introduces \PdrEV, a novel generalization of \Pdr that uses Extended Resolution (ER), a proof system exponentially stronger than Resolution, when constructing a proof of correctness. \PdrEV leverages ER to construct shorter bounded proofs of correctness, enabling it to discover more compact inductive invariants. While \PdrEV is based on \Pdr, it includes algorithmic enhancements that had to be made in order to efficiently use ER in the context of model checking.

We implemented \PdrEV in a new open-source verification framework and evaluated it on the Hardware Model Checking Competition benchmarks from 2019, 2020 and 2024. Our experimental evaluation demonstrates that \PdrEV outperforms \Pdr, solving more instances in less time and uniquely solving problems that \Pdr cannot solve within a given time limit. We argue that this paper represents a significant step toward making strong proof systems practically usable in model checking.


\end{abstract}

\input{1_intro}

\input{2_prelim}

\input{3_pdr}

\input{4_extension_rule}

\input{5_pdr_er}

\input{6_experiments}
\input{9_conclusion}

\subsection*{Acknowledgments.} This research is partially funded by the Israel Science Foundation (ISF), grant no. 2875/21.

\bibliographystyle{acm}
\bibliography{bibliography}


\end{document}

%% file: 1_intro.tex
\section{Introduction}

    Property Directed Reachability (\Pdr), a.k.a. IC3~\cite{DBLP:conf/vmcai/Bradley11,DBLP:conf/fmcad/EenMB11}, is a state-of-the-art model checking algorithm, which is widely used by many verification tools. Given a model and a safety property to be verified, \Pdr is especially effective in finding an inductive invariant that establishes the correctness of that property w.r.t. the given model. 
    Intuitively, one can view \Pdr as a proof-search algorithm, which searches for a proof, i.e. an inductive invariant, in order to establish the validity of the property.
    It is therefore natural to consider the underlying proof-system \Pdr uses during this process.

    Proof-systems are a central concept in the theory of Computer Science and are a powerful tool when applying logic in automated reasoning. A proof-system $P$ is a set of axioms and inference rules such that a derivation of a statement in $P$ is a proof for the validity of that statement under the given set of axioms. A well-known proof-system is Resolution~\cite{dp,DBLP:books/lib/Knuth97}, which forms the theoretical foundation of Boolean Satisfiability (SAT)~\cite{dp,dpll,DBLP:conf/iccad/SilvaS96,DBLP:conf/dac/MoskewiczMZZM01}. In fact, Resolution is the underlying proof-system in many SAT-based model checking algorithms~\cite{DBLP:conf/cav/McMillan03,DBLP:conf/cade/Clarke03,DBLP:conf/fmcad/VizelG09,DBLP:conf/cav/VizelG14,DBLP:journals/pieee/VizelWM15}, including \Pdr~\cite{DBLP:conf/fmcad/BaylessVBHH13}.

    In \Pdr, during the search for an inductive invariant the algorithm constructs a bounded proof of correctness, in Conjunctive Normal Form (CNF), that proves that the property holds up to a given bound. In fact, this bounded proof can be simulated by a Resolution proof~\cite{DBLP:conf/fmcad/BaylessVBHH13}. More precisely, for a bounded proof of correctness $\vF$ constructed by \Pdr, there exists a Resolution proof polynomial in the size of $\vF$. By that, the bounded proofs \Pdr constructs are limited by the strength of the Resolution proof-system.
    
    The strength of a proof-system $P$ is determined by the length of the shortest proof it admits for a valid statement. The shorter the proofs in $P$, the stronger $P$ is. In the case of Resolution, there exist formulas for which Resolution admits no ``short'' proofs, i.e. it admits no proof polynomial in the size of the formula. A notable example is the pigeonhole problem, which admits only exponential size Resolution proofs~\cite{DBLP:journals/tcs/Haken85}. Since the bounded proof of correctness \Pdr constructs can be simulated by Resolution, this means that there exist models and properties, for which \Pdr admits no short bounded proof of correctness.
    This fact limits the performance of \Pdr. Moreover, it implies that there exist models and properties that \Pdr necessarily cannot solve in a reasonable time.

    In this paper we present a novel model checking algorithm, \PdrEV, which generalizes \Pdr with Extended Resolution (ER)~\cite{Tseitin1983OnTC} as its underlying proof-system. ER is a strong proof-system 
    that is \emph{exponentially stronger} than Resolution. For example, there are polynomial size ER proofs for the pigeonhole problem~\cite{DBLP:journals/sigact/Cook76}. Hence, if \PdrEV can efficiently utilize ER, it has the potential to construct shorter bounded and unbounded proofs of correctness.

    ER generalizes Resolution with the addition of \emph{the extension rule}. It enables the addition of an auxiliary variable along with its definition to the formula. For example, given a formula $\varphi(x_1,\ldots,x_n)$, applying the extension rule may add $y \liff (x_1\lor x_2)$ to $\varphi$, where $y$ is a fresh variable. This results in the formula $\psi(x_1,\ldots,x_n,y) := \varphi(x_1,\ldots,x_n)\land (y \leftrightarrow (x_1\lor x_2))$. 

    There is a trade-off between the strength of a proof system and the ability to implement an efficient proof-search algorithm for it. Hence, while ER is exponentially stronger than Resolution, the addition of the extension rule makes proof-search algorithms for ER intractable, since there is no guidance as to which auxiliary variables and definitions should be added. 

    The complexity in implementing an efficient proof-search algorithm for ER also manifests itself in \PdrEV. Therefore, the following three key aspects must be addressed for \PdrEV to be efficient:
    \begin{enumerate*}[label=(\roman*)]
        \item the extension rule should add the ``right'' auxiliary variables and definitions; 
        \item applying the extension rule should be efficient; and
        \item the model checking algorithm needs to be adjusted to account for the added auxiliary variables.
    \end{enumerate*}
    The first two points are also true in the context of SAT. Hence, we draw our intuition from techniques that utilize ER in \satsolvers, which had some limited success~\cite{DBLP:conf/aaai/AudemardKS10,DBLP:conf/hvc/MantheyHB12}. \PdrEV uses a pre-defined set of templates that guides the addition of auxiliary variables. The templates are used to match clauses from the bounded proof maintained by \Pdr. If a set of clauses in the proof matches a given template, an auxiliary variable is added and the matched clauses are replaced with a new \emph{smaller} set of clauses.
    For example, assume the proof consists of the following clauses $(l_1\lor l_2\lor A)$ and $(\neg l_1\lor \neg l_2\lor A)$. When these clauses are identified, \PdrEV can add a fresh auxiliary variable $y$ such that $y \leftrightarrow (l_1 \lxor l_2)$, and replace the above two clauses with $(y\lor A)$, effectively re-encoding the bounded proof of correctness with the addition of an auxiliary variable. Intuitively, this is somewhat similar to ``ghost code'' addition~\cite{DBLP:conf/vmcai/ChangL05,DBLP:conf/pldi/HalbwachsP08}, only that in this case it is done \emph{automatically}.

    Using this mechanism in \PdrEV reduces the size of the bounded proof \PdrEV maintains. As it turns out, this reduction causes \PdrEV to learn redundant clauses, causing performance degradation. Hence, by preventing \PdrEV from learning redundant clauses its performance can be improved. 
    This is where the third point mentioned above becomes relevant.
    
    Using ER and adding auxiliary variables in \Pdr break many of its nice properties. For example, syntactic subsumption checks are no longer effective and may cause the proof to include many redundant clauses. This becomes a bottleneck when the number of such redundant clauses grows large. Another example comes from \emph{propagation}. If \PdrEV uses the same propagation mechanism as in \Pdr it results in \PdrEV re-learning many of the clauses that already exist in the bounded proof it maintains. Thus, while \PdrEV is based on \Pdr, in order to make it efficient, the algorithm had to be modified such that it takes into account the existence of auxiliary variables that are introduced by ER.

    We implemented \PdrEV in an open-source tool and evaluated it on the last three Hardware Model Checking Competition~\cite{DBLP:conf/fmcad/BiereDH17} benchmarks from 2019, 2020 and 2024 (HWMCC'19/20/24). We compared \PdrEV against our own implementation of \Pdr as well as the implementation of \Pdr in ABC. Our experimental evaluation shows that \PdrEV is \emph{superior} to \Pdr. It solves many problems that are not solvable by \Pdr and performs better overall in terms of runtime and number of solved instances. In addition, it produces shorter proofs on almost all instances. This leads us to conclude that \PdrEV is the first step in making strong proof systems \emph{efficiently usable} in model checking algorithms.
    
    


\subsection{Related Work}

The ER proof system received attention from the SAT community due to its strength. If \satsolvers can efficiently utilize ER, the impact can be enormous. There have been a few works that tried to incorporate ER into \satsolvers~\cite{DBLP:conf/csr/SinzB06,DBLP:conf/aaai/AudemardKS10,DBLP:journals/ai/Huang10,DBLP:conf/hvc/MantheyHB12,DBLP:conf/sat/HaberlandtGH23,DBLP:journals/tocl/BryantH23}. While these works are in the context of \satsolvers and proofs of unsatisfiability, our template-based algorithm for applying the extension rule draws intuition from~\cite{DBLP:conf/aaai/AudemardKS10,DBLP:conf/hvc/MantheyHB12}. In both these works the extension rule is applied on two clauses of the form $(l_1\lor A)$ and $(l_2 \lor A)$, driving the addition of an auxiliary variable $x\liff l_1\land l_2$. Our algorithm extends this approach and can identify other templates that can also introduce an auxiliary variable of the form $x\liff l_1\lxor l_2$. Moreover, while all of these approaches operate on a given formula in Conjunctive Normal Form (CNF) and try to speed up satisfiability checking, in the case of \PdrEV, it operates on a number of formulas in CNF that change as the algorithm makes progress when searching for an inductive invariant.

The closest work to ours is~\cite{DBLP:conf/fmcad/DurejaGIV21}. It extends \Pdr to allow it to express inductive invariants over both state variables and internal signals. While this work allows \Pdr to use auxiliary variables when searching for the inductive invariant, the set of possible auxiliary variables is limited only to these variables that appear in the verified model. In contrast, \PdrEV does not have this limitation and it can add auxiliary variables that are defined by an arbitrary Boolean function. Moreover, the algorithm presented in~\cite{DBLP:conf/fmcad/DurejaGIV21} does not perform better than \Pdr overall and has only shown better performance on a specific class of verification problems. Unlike~\cite{DBLP:conf/fmcad/DurejaGIV21}, \PdrEV outperforms \Pdr.

Another line of work that uses a similar principle is based on ``ghost code'' addition~\cite{DBLP:conf/vmcai/ChangL05,DBLP:conf/pldi/HalbwachsP08,DBLP:journals/fmsd/FilliatreGP16,DBLP:conf/vstte/BeckerM19,DBLP:conf/vmcai/ChevalierF20,DBLP:journals/pacmpl/ClochardMP20}. Ghost code augments programs with verification-specific constructs (e.g., ghost variables or assertions) that enable concise proofs without affecting runtime behavior. Intuitively, this is similar to the addition of auxiliary variables by ER.

%% file: 2_prelim.tex
\section{Preliminaries}


Given a set $U$ of Boolean variables,
a \emph{literal} $\ell$ is a variable $u\in U$ or its negation $\neg u$,
$var(l)$ denotes the associated variable of the literal $l$, and we denote by
$Lits(U)$ the set of all literals over $U$.
A \emph{clause} is a
disjunction of literals, and a formula in \emph{Conjunctive Normal
  Form} (CNF) is a conjunction of clauses. It is
convenient to treat a clause as a set of literals, and a CNF as a set
of clauses. We refer to a conjunction of literals as a \emph{cube}.

\paragraph{Safety verification:} A transition system $T$ is a tuple
$(\Vars, \Init, \Tr, \Bad)$, where $\Vars$ is a set of variables that defines the states of $T$ (i.e., 
all valuations to $\Vars$), $\Init$
and $\Bad$ are formulas with variables in $\Vars$ denoting the set
of initial states and bad states, respectively, and $\Tr$ is a formula
with free variables in $\Vars \cup \Vars'$, denoting the transition
relation. A state $s\in T$ is said to be reachable if and only if (iff)
there exists a state  
$s_0\in\Init$, and $(s_i,s_{i+1})\in\Tr$ for $0\leq i < N$, \mbox{and $s = s_N$.}

A transition system $T$ is UNSAFE iff there exists a state $s\in\Bad$
s.t. $s$ is reachable. 
When $T$ is UNSAFE and $s_N\in\Bad$ is the reachable state, the path from $s_0\in\Init$ to $s_N$ is called a \emph{counterexample} (CEX).

A transition system $T$ is SAFE iff all reachable states in $T$ do not
satisfy $\Bad$. Equivalently, there exists a formula $\Inv$, called a
\emph{safe inductive invariant}, satisfying: $\Init(\Vars) \limp Inv(\Vars)$, $\Inv(\Vars) \land \Tr(\Vars,\Vars') \limp \Inv(\Vars')$ and $\Inv(\Vars)
  \limp \neg \Bad(\Vars)$.
A \emph{safety} verification problem is to decide whether a transition
system $T$ is UNSAFE or SAFE, i.e., whether there exists an initial state
in $\Init$ that can reach a bad state in $\Bad$, or synthesize a safe
inductive invariant.

In SAT-based model checking, the verification problem is determined by
computing over-approximations of the states reachable in $T$ and, by
that, trying to either construct an invariant or find a
CEX.


\paragraph{Relative induction:} Given two formulas $F(\Vars)$ and $G(\Vars)$, $G$ is \emph{inductive relative} to $F$ iff the following two conditions hold:
\begin{align}
    Init (\Vars) \limp & G(\Vars) &
    (F(\Vars)\land G(\Vars))\land Tr(\Vars, \Vars ')  \limp G(\Vars ')
\end{align}


%% file: 3_pdr.tex
\section{Property Directed Reachability}

In this section, we give a brief overview of \Pdr~\cite{DBLP:conf/vmcai/Bradley11,DBLP:conf/fmcad/EenMB11}.
While \Pdr is well-known, we present the necessary details for the rest of the paper.
Let us fix a transition system $T = (\Vars, \Init, \Tr, \Bad)$.

The main data-structure maintained by \Pdr is a sequence of formulas $\vF=\trace$, called an \emph{inductive trace}, or simply a \emph{trace}. 
An inductive trace $\vF=\trace$ satisfies the
following two properties:
\begin{align}
  \Init(\Vars) &= F_0(\Vars) & \forall 0 \leq i < N \cdot F_i(\Vars) \land
  \Tr(\Vars,\Vars') \limp F_{i+1}(\Vars')
\end{align}
The \emph{size} of $\vF = \trace$ is $|\vF| = N$.
An element $F_i$ of a trace is called a \emph{frame}. 
The index of a frame is called a \emph{level}. $\vF$ is \emph{clausal} when all 
its elements are in CNF. We abuse notation and treat every frame $F_i$ as a set of clauses.

An inductive trace is \emph{safe} if each $F_i$ is safe: $F_i\limp\neg\Bad$;
\emph{monotone} if $\forall 0 \leq i < N, \quad F_i \limp F_{i+1}$; and
\emph{closed} if $\exists 1 \leq i \leq N
\cdot F_i \limp \left(\Lor_{j=0}^{i-1} F_j \right)$.

There is an obvious relationship between existence of safe traces and safety of a transition system:
\begin{lemma}
If a transition system
$T$ admits a safe trace $\vF$ of size $|\vF| = N$, then $T$ does not admit
counterexamples of length less than, or equal to $N$. 
\end{lemma}

We refer to such a trace as \emph{a bounded proof of correctness}. 
For an unbounded proof of correctness, the trace also needs to be closed:
\begin{theorem}
  \label{thm:closed-safety}
  A transition system $T$ is SAFE iff it admits a safe closed trace.
\end{theorem}
Thus, safety verification is reduced to searching for a safe closed trace or finding a CEX. In particular, \Pdr iteratively extends $\vF$ such that $\vF$ is safe. This procedure continues until $\vF$ is either closed or it cannot be extended and remain safe. In the latter case, \Pdr can construct a counterexample that shows why $\vF$ cannot be extended. It is important to note that the trace $\vF$ \Pdr constructs is \emph{syntactically monotone},
namely $\forall 0 < i < N, \quad F_i \supseteq F_{i+1}$. Hence, $\vF$ is closed when there exists a level $0 < i < N$ such that $F_i = F_{i+1}$. We note that in most implementations, due to the monotonicity of $\vF$, it is often represented using a \emph{delta-trace}. Given a monotone trace $\vF$ of size $N$, the \emph{delta}-trace of  $\vF$ is defined as $\vD = [D_0,\ldots, D_N]$ such that $D_0 = F_0$, $D_i = F_i\setminus F_{i+1}$ for  $0 < i < N$ and $D_N = F_N$. In what follows we use $\vF$ and $\vD$ interchangeably.

Algorithm~\ref{alg:pdr_mainLoop} presents a high level view of \Pdr's main loop. We only present the details that are required for this paper. 
\Pdr starts by initializing an inductive trace $\vF = [F_0]$ where $F_0 = \Init$ (line~\ref{alg:main:init}). The main loop iteratively extends the trace $\vF$ and tries to make it safe (line~\ref{alg:main:mk_safe}). If it fails, then a counterexample is returned (line~\ref{alg:main:cex}). Otherwise, a new frame is added to $\vF$ and it is initialized to $\top$. Once a new frame is added, \Pdr tries \emph{propagating} clauses by adding clauses that exist in the frame $F_i$ into the subsequent frame $F_{i+1}$ (line~\ref{alg:main:propagate}). If after propagation there exists a frame that equals its subsequent frame, a safe inductive invariant is found and \Pdr terminates. Otherwise, it moves to the next iteration.

\paragraph{Redundancy:} During the execution of \textsc{MkSafe()}, \Pdr adds clauses to $\vF$. When inserting a new clause $c$ into a frame $F_i\in\vF$, \Pdr removes subsumed clauses. This is done with a simple syntactic subsumption check. Namely, for every $0 < j \leq i$, if there exists a clause $d\in F_j$ such that $c\subseteq d$, $d$ is removed from $F_j$.

\paragraph{Propagation:} \Pdr performs propagation in order to construct a \emph{closed} trace. Recall that $\vF$ is \emph{syntactically monotone}. Therefore, if all clauses in $F_i \setminus F_{i+1}$ (i.e. in $D_i$) can be \emph{propagated} to $F_{i+1}$ then $\vF$ becomes closed and \Pdr terminates. Propagation is performed by checking if a clause $c\in D_i$ satisfies the following condition: $F_i\land\Tr\limp c'$. If this condition holds, $c$ can be added to $F_{i+1}$.

\paragraph{Generalization:} In a monotone
trace $\vF$, a frame $F_i\in\vF$ over-approximates the set of states reachable in up to $i$
steps of the $\Tr$. Since $\vF$ is clausal and $F_i$ is in CNF, every clause $c\in F_i$ blocks states that are \emph{necessarily unreachable} in up to $i$ steps. 
Assume a set of states is represented by the cube $\varphi$. If \Pdr identifies that the states in $\varphi$ are unreachable in up to $i$ steps, it can add the clause $c = \neg\varphi$ to $F_i$. This means that the formula $(F_{i-1}\land c)\land\Tr\limp c'$ is valid. \Pdr uses this fact and tries to construct a \emph{stronger clause} $d\subseteq c$ such that $\Init\limp d$ and $(F_{i-1}\land d)\land\Tr\limp d'$ are valid. By that, it can deduce that a larger set of states is unreachable in up to $i$ steps. \Pdr refers to this process as \emph{inductive generalization} as it is based on the fact that $c$ is \emph{relatively inductive} w.r.t. $F_{i-1}$.

\begin{figure}[h]
\vspace{-8pt}
    \centering

    \begin{adjustbox}{scale=0.7}
    \begin{minipage}[t]{7.5cm}
    \begin{algorithm}[H]
    \caption{\Pdr}
    \label{alg:pdr_mainLoop}
    \KwIn{Transition System}
    \KwOut{Proof or counterexample}
    \SetKw{Initially}{Initially:}
    \Initially{$\vF = [\Init]$\label{alg:main:init}}
    
    \While{true} {
        $\vF$ = \textsc{MkSafe($\vF$)}\;\label{alg:main:mk_safe}
        \If {$\neg\vF$.\textsc{IsSafe()}} {
           \Return cex\;\label{alg:main:cex}
        }
        $\vF$.\textsc{Extend($\top$)}\;
        $\vF$.\textsc{Propagate()}\;\label{alg:main:propagate}
        $N$ = $\vF$.\textsc{Size()}\;
        \If {$\exists 0\leq i < N\cdot F_i = F_{i+1}$} {
            \Return $F_i$\;\label{alg:main:inv}
        }
        
    }
    \end{algorithm}
    \end{minipage}
    \end{adjustbox}
    \hfill
    \begin{adjustbox}{scale=0.7}
    \begin{minipage}[t]{7.5cm}
    \begin{algorithm}[H]
    \caption{\PdrEV}
    \label{alg:pdrer_mainLoop}
    \KwIn{Transition System}
    \KwOut{Proof or counterexample}
    \SetKw{Initially}{Initially:}
    \Initially{$\vG = [\Init]$\label{alg:mainer:init}}

    size = 0\;
    \While{true} {
        $\vG$ = \textsc{MkSafe($\vG$)}\;\label{alg:mainer:mk_safe}
        \If {$\neg\vG$.\textsc{IsSafe()}} {
           \Return cex\;\label{alg:mainer:cex}
        }
        \If {$\vG.\textsc{NumClauses}() > \text{size} + \Delta$ } {\label{alg:mainer:heuristic_begin}
            $\vG.\textsc{ReEncode}()$\;
            size = $\vG.\textsc{NumClauses}()$\;\label{alg:mainer:heuristic_end}
        }
        $\vG$.\textsc{Extend($\top$)}\;
        $\vG$.\textsc{Propagate()}\;\label{alg:mainer:propagate}
        $N$ = $\vG$.\textsc{Size()}\;
        \If {$\exists 0\leq i < N\cdot G_i = G_{i+1}$} {
            \Return $G_i$\;\label{alg:mainer:inv}
        }
        
    }
    \end{algorithm}
    \end{minipage}
    \end{adjustbox}
    \vspace{-25pt}
\end{figure}

%% file: 4_extension_rule.tex
\section{Applying the Extension Rule}\label{sec:apply_er}

In this section we describe how we apply the \emph{extension rule} in \PdrEV.

The goal of using ER in \PdrEV is to enable it to construct shorter bounded proofs of correctness, and by that, construct a smaller inductive invariant. When applying the extension rule, \PdrEV introduces auxiliary variables along with their definitions, and by that re-encodes the trace it maintains. We draw intuition from \satsolvers that take advantage of ER~\cite{DBLP:conf/aaai/AudemardKS10,DBLP:conf/hvc/MantheyHB12,DBLP:conf/sat/HaberlandtGH23}.

Our algorithm is based on \emph{template-matching}. It uses a pre-defined set of templates and tries to match clauses in the trace against these templates. When a template is matched, a new auxiliary variable is added along with its definition, and the clauses that match the template are replaced with new clauses, which contain the auxiliary variable. Next we give formal definitions for a template and a matched set of clauses.

Let $\mathcal{V}$ be a set of template variables (disjoint from variables of the transition system $T$). A template is defined as $\tau\subseteq \mathcal{P}(Lits(\mathcal{V}))$. Let $C$ be a set of clauses, the literals appearing in all clauses in $C$ are denoted as $L(C) = \cup C$.

\begin{definition}[Match]\label{def:match}
    Let $C$ be a set of clauses and $\tau\subseteq\mathcal{P}(Lits(\mathcal{V}))$ a template.
    We say that $C$ matches $\tau$ iff there exists an injective substitution $\sigma: \mathcal{V} \rightarrow L(C)$ and a bijective function $m : \tau \rightarrow C$ such that:
    \begin{enumerate}
        \item Match rule: $\forall M\in \tau,\exists c\in C\cdot (m(M) = c) \land (M[\sigma]\subseteq c)$ 
        \footnote{$M[\sigma] = \{ \sigma(l) | l \in M, l = var(l)\} \cup \{ \neg \sigma(var(l)) | l \in M, l = \neg var(l) \}$}

        \item Exclusivity rule: $\forall M_1,M_2\in \tau,\forall c_1,c_2\in C \cdot (m(M_1) = c_1 \land m(M_2) = c_2)\limp ((c_1\setminus M_1[\sigma]) = (c_2\setminus M_2[\sigma]))$
    \end{enumerate}
\end{definition}

For simplicity of presentation, let us define the set of template variables as $\mathcal{V} = \{\alpha, \beta, \gamma, \delta\}$. We use the following three templates when applying the extension rule: 
\begin{enumerate*}[label=(\arabic*)]
    \item \andt-template: $\tau_{\land} = \{\{\alpha\},\{\beta\}\}$;
    \item \xort-template: $\tau_{\lxor} = \{\{\alpha, \beta\},\{ \neg\alpha, \neg\beta\}\}$;
    \item \ha-template: $\tau_{ha} = \{\{\alpha,\beta,\gamma\},\{\alpha,\beta,\delta\},\{\neg\alpha,\neg\beta,\gamma,\delta\}\}$
\end{enumerate*}

\begin{definition}[Instantiation]\label{def:instantiation}
    Given a set of clauses $C$, assume $C$ matches $\tau_{*}$ where $*\in\{\land,\lxor,ha\}$. The \emph{instantiation of $\tau_*$} is:
    \begin{enumerate}
        \item $\sigma[\tau_{\land}] = \{\{\sigma(\alpha)\},\{\sigma(\beta)\}\}$ when $\tau_* = \tau_{\land}$.
        \item $\sigma[\tau_{\lxor}] = \{\{\sigma(\alpha),\sigma(\beta)\},\{\neg\sigma(\alpha),\neg\sigma(\beta)\}\}$ when $\tau_* = \tau_{\lxor}$.
        \item $\sigma[\tau_{ha}] = \{\{\sigma(\alpha),\sigma(\beta),\sigma(\gamma)\},\{\sigma(\alpha),\sigma(\beta),\sigma(\delta)\},\{\neg\sigma(\alpha),\neg\sigma(\beta),\sigma(\gamma),\sigma(\delta)\}\}$ when $\tau_* = \tau_{ha}$.
    \end{enumerate}
\end{definition}

The \andt-template and \xort-template require two clauses for a match. The clauses should be of the following form: $\{(\alpha \lor A), (\beta\lor A)\}$ for $\tau_{\land}$ (two clauses that are similar except for one literal) and $\{(\alpha\lor\beta \lor A), (\neg\alpha\lor \neg\beta\lor A)\}$ for $\tau_{\lxor}$.
When a set of clauses is matched with either $\tau_{\land}$ or $\tau_{\lxor}$, a new auxiliary variable $x$ can be added such that 
$x\liff \sigma(\alpha)\land \sigma(\beta)$ or $x\liff \sigma(\alpha)\lxor\sigma(\beta)$, respectively. Note that $\sigma$ is the substitution function for the match (see Definition~\ref{def:match}). The matched clauses are then replaced with the clause $(x \lor A)$.


\begin{figure}
\centering

\begin{adjustbox}{scale=0.75}

\begin{subfigure}
{0.3\textwidth}
\centering
\begin{math}
\begin{aligned}[c]
(a \lor \textcolor{red}{l_1 \lor l_2}) \\
(b \lor \textcolor{red}{l_1 \lor l_2}) \\
(a \lor \textcolor{blue}{l_3}) \\
(b \lor \textcolor{blue}{l_3}) \\
(a \lor \textcolor{OliveGreen}{l_4 \lor l_5\lor l_6}) \\
(b \lor \textcolor{OliveGreen}{l_4 \lor l_5\lor l_6}) \\
\downarrow \\ 
(x \lor \textcolor{red}{l_1 \lor l_2}) \\
(x \lor \textcolor{blue}{l_3}) \\
(x \lor \textcolor{OliveGreen}{l_4 \lor l_5\lor l_6}) \\
(x \liff a\land b)
\end{aligned}
\end{math}
\caption{$\tau_{\land}$}\label{fig:and}
\end{subfigure}
\begin{subfigure}
{0.3\textwidth}
\centering
\begin{math}
\begin{aligned}[c]
(\neg a \lor b  \lor \textcolor{red}{l_1 \lor l_2}) \\
(a \lor \neg b \lor \textcolor{red}{l_1 \lor l_2})  \\
(a \lor \neg b \lor \textcolor{blue}{l_3}) \\
(\neg a \lor b  \lor \textcolor{blue}{l_3}) \\
(a \lor b \lor \textcolor{OliveGreen}{l_4 \lor l_5 \lor l_6}) \\
(\neg a \lor \neg b \lor \textcolor{OliveGreen}{l_4 \lor l_5\lor l_6}) \\
\downarrow \\ 
(\neg x \lor \textcolor{red}{l_1 \lor l_2}) \\
(\neg x \lor \textcolor{blue}{l_3}) \\
(x \lor \textcolor{OliveGreen}{l_4 \lor l_5\lor l_6}) \\
(x \liff a\lxor b)
\end{aligned}
\end{math}
\caption{$\tau_{\lxor}$}\label{fig:xor}
\end{subfigure}
\begin{subfigure}
{0.3\textwidth}
\centering
\begin{math}
\begin{aligned}[c]
(a \lor b \lor c \lor \textcolor{red}{l_1 \lor l_2}) \\
(a \lor b \lor d \lor \textcolor{red}{l_1 \lor l_2}) \\
(\neg a \lor \neg b \lor c \lor d \lor \textcolor{red}{l_1 \lor l_2}) \\
\downarrow \\ 
(x \lor y \lor d \lor \textcolor{red}{l_1 \lor l_2}) \\
(x \lor z \lor c \lor \textcolor{red}{l_1 \lor l_2}) \\
(x \liff a\lxor b) \\
(y \liff a\land b) \\ 
(z \liff y\land d)
\end{aligned}
\end{math}
\caption{$\tau_{ha}$}\label{fig:halfadder}
\end{subfigure}

\end{adjustbox}

\caption{Clauses at the top match $\tau_{\land}$, $\tau_{\lxor}$ and $\tau_{ha}$ from left to right, respectively. Applying the extension rule transform them into the clauses at the bottom. The instantiations are $\{\{a\},\{b\}\}$ for $\tau_{\land}$, $\{\{ a,b\},\{\neg a, \neg b\}\}$ for $\tau_{\lxor}$ and $\{\{ a, b, c\}, \{ a, b, d\}, \{ \neg a, \neg b,c,d \} \}$ for $\tau_{ha}$.
}
\label{fig:bva_patterns}
\vspace{-12pt}
\end{figure}

The third template, $\tau_{ha}$, is more complex. Unlike $\tau_{\land}$ and $\tau_{\lxor}$, this template requires three clauses of the following form: $\{(\alpha\lor\beta\lor\gamma\lor A), (\alpha\lor\beta\lor\delta\lor A), (\neg\alpha\lor\neg\beta\lor\gamma\lor\delta\lor A)\}$.
When a set of clauses is matched with $\tau_{ha}$, three new auxiliary variables $x, y, z$ can be added such that 
$x \liff \sigma(\alpha) \lxor \sigma(\beta)$,
$y \liff \sigma(\alpha) \land \sigma(\beta)$, and
$z \liff y \land \sigma(\delta)$.
The matched clauses are then replaced with the clauses $(x \lor y \lor \delta \lor A)$, $(x \lor z \lor \gamma \lor A)$.
Figure~\ref{fig:bva_patterns} presents examples of sets of clauses that match $\tau_{\land}$, $\tau_{\lxor}$ and $\tau_{ha}$, respectively. 

\subsection{Matching Templates on an Inductive Trace}

In this section we describe the algorithm that analyzes the inductive trace \Pdr maintains in order to match templates and apply the extension rule.

Using ER as part of \Pdr is different from using ER in SAT for the following reasons:
\begin{enumerate*}[label=(\roman*)]
    \item there are multiple formulas in CNF that need to be considered (a trace $\vF$ contains multiple formulas);
    \item the CNF formulas change throughout the run of the algorithm; and
    \item the operations performed by \Pdr during its execution must be taken into account.
\end{enumerate*}

\paragraph{Multiple CNFs:} When an auxiliary variable is added to $\vF$, a set of clauses is re-encoded and replaced with other clauses. As an example, consider two clauses $c_1 = (a\lor b \lor l_1)$ and $c_2 = (\neg a\lor \neg b\lor l_1)$ that match $\tau_{\lxor}$. A new auxiliary variable $x\liff a\lxor b$ is added, and a new clause $c = (x \lor l_1)$ replaces $c_1$ and $c_2$. Note that $c$ represents \emph{a set of clauses}, namely it replaces the conjunction $c_1\land c_2$. Now, let us assume that $c_1,c_2\in F_i$ while $c_2\not\in F_{i+1}$. This could be due to \Pdr \emph{propagating} $c_1$ from level $i$ to $i+1$, while failing to propagate $c_2$ from level $i$ to $i+1$. If we consider $\vF$ as the union of all clauses that appear in it and perform matching on this union, then $x$ and $c$ are used to re-encode $\vF$ as a whole. In this case $c_1$ and $c_2$ are removed from $\vF$, and therefore $c_1$ must be removed from $F_{i+1}$. Such an operation causes information loss. In order to avoid such a scenario, when ER is applied to $\vF$, each frame is considered independently. Since $\vF$ is monotone, the algorithm analyzes only the difference between the frames, namely the delta-trace. 

\paragraph{Evolving CNF:} Since the inductive trace $\vF$ changes during the execution of \Pdr, the addition of auxiliary variables must not be aggressive. Consider the case where an auxiliary variable $x$ is added and $\vF$ is re-encoded, resulting in $\vF^x$. Later, \Pdr continues and adds the set of clauses $C$ to $\vF^x$. Some templates may not match on $\vF^x\cup C$, while they can be matched if the matching algorithm is applied on $\vF\cup C$. This is due to the fact that template matching is \emph{syntactic}. Hence, we apply re-encoding periodically based on a heuristic that inspects the number of clauses in $\vF$ (see Section~\ref{sec:pdrer}).

\paragraph{Matching Algorithm:} Matching the given templates can be done in polynomial time. A na\"{i}ve implementation tries all combinations of two ($\tau_{\land}$ and $\tau_{\lxor})$ or three ($\tau_{ha}$) clauses, resulting in $O(n^2)$ or $O(n^3)$ complexity, respectively. However, such an implementation is time consuming, and in practice, not efficient. 

In order to efficiently implement the template matching algorithm, we partition the clauses in a CNF into buckets, where each bucket contains clauses of a similar size (i.e. with the same number of literals). This is based on the following observation. In order for two clauses $c_1$ and $c_2$ to match $\tau_{\land}$ and $\tau_{\lxor}$, they must be of the same size ($|c_1|=|c_2|$). In addition, for three clauses $c_1,c_2$ and $c_3$ to match $\tau_{ha}$, it must hold that $|c_1|=|c_2|=|c_3|-1$. 
The template matching algorithm uses these conditions. It searches in a given bucket when trying to match $\tau_{\land}$ and $\tau_{\lxor}$, and searches in two consecutive buckets when trying to match $\tau_{ha}$.  

\begin{figure}[h!]
\vspace{-10pt}
\centering

\begin{adjustbox}{scale=0.7}
    \begin{minipage}[t]{8.6cm}
\begin{algorithm}[H]
\caption{\textsc{ReEncode}}\label{alg:reencode}
\KwIn{$\vD$, a delta-trace}
\KwOut{a delta-trace}
    \For {$D_i\in \Vec{D}$} {
        $M = M\cup \textsc{MatchT}(D_i)$\;\label{fig:alg:reencode:match}
    }


    \tcc{Cluster matches by their instantiations}
    $clusters = \textsc{HashMap}()$\;\label{fig:alg:reencode:inst_begin}
    \For{$m \in M$} {
        $key = m.\textsc{Instantiation()}$\;
        $clusters[key] = clusters[key] \cup m$\;
    }\label{fig:alg:reencode:inst_end}
    
    \tcc{Choose best template to use}
    $key = choose(clusters)$ \;\label{fig:alg:reencode:choose}
    \If{$key == None$} {
        \Return {$\vec{D}$}\;
    }

    \tcc{Perform the re-encoding}
    \For{$m \in clusters[key]$} {\label{fig:alg:reencode:reencode_begin}
        $D_i = m.\textsc{GetCNF()}$\;
        $D_i = D_i \setminus m.\textsc{GetMatchedClauses()}$\;
        $D_i = D_i \cup m.\textsc{ApplyExtension()}$\;
    }\label{fig:alg:reencode:reencode_end}
\Return {$\vec{D}$}\;
\end{algorithm}
        
    \end{minipage}
\end{adjustbox}
\hfill
\begin{adjustbox}{scale=0.7}
    \begin{minipage}[t]{6.8cm}
    \begin{algorithm}[H]
\caption{\ImpEV}\label{alg:imp}
\label{alg:pdr_er_implies}

\KwIn{$c_1(\bar v, \bar a)$, $c_2(\bar v, \bar a)$: clauses}
\KwOut{Boolean \emph{true}/\emph{false}}
\Ensures{Returns \emph{true} iff $c_1 \limp c_2$}

\If{$c_1 \subseteq c_2$} { \label{alg:implication:subsumption}
    \Return \emph{true}\;
}
\If {$(c_1 \cup c_2) \cap Lits(\bar a) == \emptyset $} { \label{alg:implication:no_aux}
    \Return \emph{false}\;
}

\If{$\exists l\in c_1 \cdot coi(l)\cap coi(c_2) = \emptyset$} { \label{alg:implication:coi_begin}
\Return \emph{false}\;\label{alg:implication:coi_end}
}

$bdd_1 = cachedBuildBDD(c_1)$\;
$bdd_2 = cachedBuildBDD(c_2)$\;
\Return $(bdd_1 \limp bdd_2) = \top$\;\label{alg:implication:bdd}

\end{algorithm}
    \end{minipage}
\end{adjustbox}
\vspace{-15pt}
 \end{figure}

The overall procedure for adding auxiliary variables and re-encoding $\vF$ appears in Algorithm~\ref{alg:reencode}. 
The algorithm starts by iterating over all elements of the delta-trace $\vD$, and applying the matching algorithm (line~\ref{fig:alg:reencode:match}). Once all templates are identified across the different frames, the algorithm tries to detect templates that correspond to multiple frames. By that, an added auxiliary variable can be used in multiple frames and as a result, a better reduction in the number of overall clauses is achieved. The algorithm clusters all matches using their template instantiation (lines~\ref{fig:alg:reencode:inst_begin}-\ref{fig:alg:reencode:inst_end}).
Once all matches are clustered, the algorithm chooses a match according to a pre-defined heuristic (line~\ref{fig:alg:reencode:choose}). For example, throughout our experiments we noticed the $\tau_{\lxor}$ is usually preferable over other templates, even when other templates may result in a more significant reduction in the number of clauses in $\vD$.
Lastly, the algorithm iterates over all matches for the selected instantiation and performs the re-encoding of $\vD$. More precisely, it adds a new auxiliary variable along with its definition, and replaces the corresponding clauses according to the type of template (lines~\ref{fig:alg:reencode:reencode_begin}-\ref{fig:alg:reencode:reencode_end}).

Note that while the presented algorithm only uses one instantiation, in practice we can choose a number of instantiations and thus introduce many new auxiliary variables at each invocation of \textsc{ReEncode}.

%% file: 5_pdr_er.tex
\section{\PdrEV}\label{sec:pdrer}

In this section we describe a new model checking algorithm, \PdrEV, which is based on \Pdr. \PdrEV uses ER to re-encode the trace and includes other algorithmic modifications that take into account the added auxiliary variables. We emphasize that simply using ER to re-encode the trace does not use the full power of ER.
Moreover, once the auxiliary variables are added, many of the nice features of \Pdr break. For example, syntactic subsumption is no longer effective, this results in redundant clauses that slow down \Pdr. Before describing the details of \PdrEV we generalize the notion of an invariant and an inductive trace to take into account the addition of auxiliary variables.

\begin{definition}[Auxiliary Circuit]\label{def:aux_circ}
    Let $T=(\Vars,\Init,\Tr,\Bad)$ be a transition system and let $\bar a = \{a_1,\ldots,a_m\}$ be a set of auxiliary variables. The \emph{auxiliary circuit} for $\bar a$ is the formula $E(\Vars,\bar a) = \bigwedge\limits_{i=1}^{m} (a_i\liff l_{i_1}\star_i l_{i_2})$ such that $$\forall 1 \leq i \leq m\cdot (\star_i\in \{\land,\lxor\}) \land (l_{i_1}, l_{i_2} \in   Lits(\bar v \cup \{ a_j | a_j \in \bar a \land j < i\}))$$
\end{definition}

We define the cone-of-influence (COI) of a literal $l\in Lits(\Vars\cup\bar a)$ as
\begin{equation}
    coi(l) =
    \begin{cases}
      var(l) &  var(l)\in\Vars \\
      coi(l_1)\cup coi(l_2) &  var(l)\in\bar a \land var(l) \liff l_1\star l_2
    \end{cases}
  \end{equation}

Intuitively, the COI includes all the variables in $\Vars\cup\bar a$ that are used to define $l$ in $E(\Vars,\bar a)$.
  
From this point on, let us fix a transition system $T=(\Vars,\Init,\Tr,\Bad)$, a set of auxiliary variables $\bar a = \{a_1,\ldots,a_m\}$ and its auxiliary circuit $E(\Vars, \bar a)$.

When using ER and auxiliary variables, a safe inductive invariant is no longer expressed only in terms of $\Vars$, but also includes the added auxiliary variables. We therefore generalize the definition of a safe inductive invariant:

\begin{definition}[Generalized Invariant]\label{def:gen_inv} 
    A \emph{generalized safe inductive invariant} is a formula $\Inv(\Vars,\bar a)$ (over $\Vars\cup\bar a$) that satisfies:
\begin{align}
Init(\bar v) \land E(\bar v, \bar a) &\limp Inv(\bar v, \bar a)\label{eq:gen_inv:init} \\
Inv(\bar v, \bar a) \land E(\bar v, \bar a) \land Tr(\Vars,\Vars') \land E(\bar v', \bar a') &\limp Inv(\bar v', \bar a')\label{eq:gen_inv:cons} \\ 
Inv(\bar v, \bar a) \land E(\bar v, \bar a) &\limp \neg\Bad(\Vars)\label{eq:gen_inv:safe}
\end{align}

\end{definition}

\begin{lemma}\label{lemma:satisfaction}
    Given a formula $\varphi(\Vars,\bar a)$ over $\Vars\cup \bar a$ and assume $a_n \liff l_1\star l_2$ where $\star\in \{\land,\lxor\}$. Let $\bar a^* = \bar a\setminus\{a_n\}$ and $E^*(\Vars,\bar a^*)$ be its auxiliary circuit such that $E(\Vars, \bar a) = E^*(\Vars,\bar a^*)\land (a_n \liff l_1\star l_2)$. Let us define $\psi(\Vars, \bar a^*) := \varphi[a_n\gets l_1\star l_2]$, namely, every occurrence of $a_n$ in $\varphi$ is substituted by its definition $l_1\star l_2$. Given an assignment $z^*:\Vars\cup\bar a^*\rightarrow \{0,1\}$ over $\Vars\cup \bar a ^*$ and an assignment $z:\Vars\cup\bar a\rightarrow \{0,1\}$ over $\Vars\cup\bar a$, which satisfy the following conditions: $z(x) = z^{*}(x)$ for $x\in \Vars\cup\bar a^*$ and $z(a_n) = z^{*}(l_1)\star z^{*}(l_2)$. Then, $z\models \varphi\land E(\Vars,\bar a)$ if and only if $z^*\models \psi\land E^*(\Vars,\bar a^*)$.
\end{lemma}

\begin{proof}
    The proof is immediate from the structure of the formulas. $z\models \varphi(\Vars\cup\bar a)\land E(\Vars\cup\bar a)$ iff $z\models \varphi(\Vars\cup\bar a)\land E^*(\Vars\cup\bar a^*)\land a_n \liff l_1\star l_2$, iff $z\models \varphi(\Vars\cup\bar a)$ and $z\models E^*(\Vars\cup\bar a^*)$ and $z\models a_n\liff l_1\star l_2$. Note that since $z$ and $z^*$ agree on $\Vars\cup\bar a^*$, and by the fact $E^*$ is only over $\Vars\cup\bar a^*$, we get that $z\models E^*(\Vars\cup\bar a^*)$ iff $z^*\models E^*(\Vars\cup\bar a^*)$. Hence we only consider $z\models \varphi\land a_n\liff l_1\star l_2$. 
    
    We prove that $z\models \varphi(\Vars\cup\bar a)\land a_n\liff l_1\star l_2$ iff $z^*\models \psi(\Vars\cup \bar a^*)$ by induction on the structure of $\varphi$. For the base case, $\varphi$ is simply a variable. If that variable is in $\Vars\cup\bar a^*$, we immediately conclude $z^*\models \psi$ since $\psi = \varphi$ and $z$ agrees with $z^*$ on $\Vars\cup\bar a^*$. If the variable is $a_n$ then $\psi = l_1\star l_2$, and by definition and the fact $z\models a_n\liff l_1\star l_2$ we get that $z^*\models \psi$. For the induction step, assume $\varphi = \varphi_1 \land \varphi_2$. By the induction hypothesis $z\models \varphi_i$ iff $z^*\models \varphi_i[a_n\gets l_1\star l_2]$, for $i\in\{1,2\}$. Hence, it immediately follows that $z^*\models \varphi_1[a_n\gets l_1\star l_2]\land \varphi_2[a_n\gets l_1\star l_2]$, i.e. $z^*\models \psi$. A similar proof is applicable for other connectives, such as $\neg, \lor, \lxor$ etc.
    \qed
\end{proof}

\begin{lemma}[Auxiliary Variable Elimination]\label{lemma:av_elimination}
If $\Inv_1(\Vars,\bar a)$ is a generalized safe inductive invariant w.r.t. $T$, where $\bar a = \{a_1,\ldots,a_n\}$ for $n\geq 1$, then there exists a generalized safe inductive invariant $\Inv_2(\Vars,\bar a^*)$ where $\bar a^* = \bar a\setminus\{a_n\}$.
\end{lemma}
\begin{proof}[Sketch]
    Assume $a_n \liff l_1\star l_2$ where $\star\in \{\land,\lxor\}$. Let $E^*(\Vars,\bar a^*)$ be the auxiliary circuit for $\bar a^*$, and let $E(\Vars, \bar a) = E^*(\Vars,\bar a^*)\land (a_n \liff l_1\star l_2)$ be the auxiliary circuit for $\bar a$. Let us define $\Inv_2(\Vars,\bar a^*) = \Inv_1(\Vars,\bar a)[a_n\gets l_1\star l_2]$. We need to show that $Inv_2$ is a generalized safe inductive invariant, namely that Equations~\ref{eq:gen_inv:init}-\ref{eq:gen_inv:safe} hold for $Inv_2$. 
    This can be shown using Lemma~\ref{lemma:satisfaction}.
    \qed
    
\end{proof}

\begin{lemma}\label{lemma:generalized_inv_to_inv}
    If $\Inv_g(\Vars,\bar a)$ is a generalized safe inductive invariant w.r.t. $T$, then there exists a safe inductive invariant $\Inv(\Vars)$ w.r.t. $T$.
\end{lemma}
\begin{proof}
    By induction on the size of $\bar a$. If $|\bar a|=0$ then $\bar a = \emptyset$ and $\Inv_g(\Vars, \bar a)$ is a formula over $\Vars$. Therefore, $Inv(\Vars) = \Inv_g(\Vars,\bar a)$ is a safe inductive invariant.

    For the induction step, assume that for a generalized inductive invariant  $\Inv_g^*(\Vars,\bar a^*)$, where $|\bar a^*|=n$, there exists a safe inductive invariant $\Inv^{\bar a^*}(\Vars)$. Next, assume $\Inv_g(\Vars,\bar a)$ is a generalized safe inductive invariant where $\bar a = \{a_1,\ldots,a_{n+1}\}$. Clearly, $|\bar a|=n+1$. Note that $a_{n+1}\in\bar a$ is the last auxiliary variable in $\bar a$ (recall that $\bar a$ is ordered). Let $\bar a^* = \bar a \setminus\{a_{n+1}\}$, by Lemma~\ref{lemma:av_elimination} there exists a generalized safe inductive invariant $\Inv_g^*(\Vars,\bar a^*)$.
    By the induction hypothesis, there exists a safe inductive invariant $\Inv^{\bar a^*}(\Vars)$.\qed
\end{proof}

Since \PdrEV maintains an inductive trace like \Pdr, we also generalize the definition of an inductive trace:
\begin{definition}[Generalized Inductive Trace]
    A \emph{generalized inductive trace} $\vG = [G_0,\ldots,G_N]$ satisfies the following properties:
    \begin{align}
  \Init(\Vars) \land E(\Vars, \bar a) &\limp G_0(\Vars,\bar a) \\ 
  \forall 0 \leq i < N \cdot G_i(\Vars,\bar a)\land E(\Vars,\bar a) \land
  \Tr(\Vars,\Vars')\land E(\Vars',\bar a') &\limp G_{i+1}(\Vars', \bar a')
\end{align}
\end{definition}

\begin{theorem}
    $T$ is SAFE iff it admits a safe closed \emph{generalized} trace.
\end{theorem}
\begin{proof}
    Assume $T$ is SAFE. By Theorem~\ref{thm:closed-safety} $T$ admits a safe closed trace $F$. $F$ is a generalized trace where $\bar a = \emptyset$ and $E(\Vars, \bar a)=\top$. Hence, this direction is trivial.

    Assume $T$ admits a safe closed generalized trace $G$ of size $N$. Since $G$ is closed, there exists $1\leq j < N$ such that $G_j = G_{j+1}$, and $G_j$ is a generalized safe inductive invariant. By Lemma~\ref{lemma:generalized_inv_to_inv} there exists a safe inductive invariant for $T$, and hence $T$ is SAFE.
    \qed
\end{proof}

Note that relative induction and implication checks are adjusted accordingly. For example, checking if a clause $c\in G_i$ can be propagated to $G_{i+1}$ is done by checking if the following formula is valid $$G_i(\Vars,\bar a)\land E(\Vars,\bar a)\land \Tr(\Vars,\Vars')\land E(\Vars',\bar a')\limp c(\Vars',\bar a')$$

\PdrEV maintains a generalized trace $\vG$. Its main loop is similar to the main loop of \Pdr and appears in Algorithm~\ref{alg:pdrer_mainLoop}. The only difference in the main loop is that $\vG$ is periodically re-encoded by applying the extension rule and adding auxiliary variables (as described in Section~\ref{sec:apply_er}). The frequency of re-encoding is based on a heuristic captured by the parameter $\Delta$ (lines~\ref{alg:mainer:heuristic_begin}-\ref{alg:mainer:heuristic_end}). As we mentioned above, this change alone can reduce the number of clauses in $\vG$ but this reduction does not necessarily translate to better performance overall. 
Next, we describe the algorithmic modifications that separate \PdrEV from \Pdr.

There are three main differences between \PdrEV and \Pdr due to the fact that $\vG$ is over $\Vars\cup \bar a$:
\begin{enumerate*}[label=(\roman*)]
    \item clause redundancy checks cannot be syntactic;
    \item auxiliary variables are explicitly used for generalization; and
    \item clauses that include auxiliary variables and fail to propagate must be handled differently.
\end{enumerate*}

\subsection{Redundant Clauses}

In \Pdr, when a new clause $c$ is added to $F_i$, a subsumption check is performed. Every clause $d\in D_j$, for $0< j\leq i$ , subsumed by $c$ is removed. Since in \Pdr $\vF$ is only over state variables, this subsumption check is performed by checking if $c\subseteq d$ holds. However, for a generalized trace like the one maintained by \PdrEV, this simple check does not suffice. For example, consider the clauses $c_1 = y\lor l_1$ and $c_2 = a \lor l_1\lor l_2$, where $y$ is an auxiliary variable defined as $y\liff a\land b$. Clearly, $c_1\limp c_2$, but a simple syntactic subsumption check fails to identify this fact. Hence, a semantic implication check must be performed. In order to perform efficient implication checks we use Binary Decision Diagrams (BDDs)~\cite{DBLP:reference/mc/Bryant18}.

Algorithm~\ref{alg:imp} presents how \PdrEV checks for implication between two clauses. It starts by performing the standard (syntactic) subsumption check (line~\ref{alg:implication:subsumption}). Then, it checks if the two clauses contain auxiliary variables (line~\ref{alg:implication:no_aux}). If they do not contain auxiliary variables and the initial subsumption check failed, it returns false. 
If the clauses contain auxiliary variables, the algorithm checks if there exists a literal $l\in c_1$ such that its COI does not intersect the COI of $c_2$ (lines~\ref{alg:implication:coi_begin}-\ref{alg:implication:coi_end}). 
Note that if there exists such a literal $l\in c_1$, an assignment that satisfies $c_1$ but not $c_2$ can be constructed. Hence implication does not hold and the algorithm returns false. Lastly, if all syntactic checks fail, BDDs are used for the implication check. If the BDD representing $c_1\limp c_2$ is evaluated to $\top$, the algorithm returns true; otherwise, it returns false.
Note that a clause over $\Vars \cup \bar a$ corresponds to a Boolean function over $\Vars$. Hence, BDD construction for such a clause may be computationally expensive. Since implication checks are performed frequently, \ImpEV caches BDDs to avoid reconstructing the same BDD multiple times.

We emphasize that the use of BDDs in this case is crucial. As it turns out, trying to use a \satsolver for these checks results in decreased performance.

\begin{figure}[h!]
\vspace{-10pt}
\centering






\begin{adjustbox}{scale=0.7}
    \begin{minipage}[t]{9.4cm}





\begin{algorithm}[H]
\caption{\GenEV}\label{alg:gen}
\SetKw{Continue}{continue}
\KwIn{$\varphi(\Vars)$: cube, i: int}
\KwOut{$c(\Vars,\bar a)$: clause}
\Requires{$G_{i-1}(\Vars,\bar a)\land E(\Vars,\bar a)\land\Tr(\Vars,\Vars')\land E(\Vars',\bar a')\limp\neg\varphi'$}
\Ensures{$G_{i-1}(\Vars,\bar a)\land E(\Vars,\bar a)\land c\land\Tr(\Vars,\Vars')\land E(\Vars',\bar a')\limp c'$}
$c(\bar v) = \textsc{IndGen}(i, \neg\varphi)$\;\label{fig:alg:gener:ind_gen}

\For {$a \in \bar a$ where $a \liff l_1 \star l_2$} {
    $d = None$\;
    \If {$(\star = \land)$} {\label{alg:gener:and_begin}
        \lIf{$(l_1 \in c) \land (l_2 \not\in c)$} {
            $d = (c \setminus \{ l_1 \}) \cup \{ a \}$
        }
        \lElseIf {$(l_1 \not\in c) \land (l_2 \in c)$} {
            $d = (c \setminus \{ l_2 \}) \cup \{ a \}$\label{alg:gener:and_end}
        }
    } 
    \If {$(\star = \lxor)$} {\label{alg:gener:xor_begin}
        \If {$(l_1 \in c) \land (\neg l_2 \in c)$} {
            $d = (c \setminus \{l_1,\neg l_2\}) \cup \{ \neg a \}$
        }
        \ElseIf{$(\neg l_1 \in c) \land (l_2 \in c))$} {
            $d = (c \setminus \{\neg l_1, l_2\}) \cup \{ \neg a \}$
        }
        \ElseIf{$(l_1 \in c) \land (l_2 \in c))$} {
            $d = (c \setminus \{ l_1, l_2\}) \cup \{ a \}$
        }
        \ElseIf{$(\neg l_1 \in c) \land (\neg l_2 \in c))$} {
            $d = (c \setminus \{\neg l_1, \neg l_2\}) \cup \{  a \}$ \label{alg:gener:xor_end}
        }
    }

    \lIf{$d = None$} {\Continue}

    \tcc{Is $d$ inductive relative to $G_{i-1}$?}
    \If{$Init(\Vars) \land E(\Vars, \bar a) \limp d(\Vars,\bar a)$} {\label{alg:genev:rel_ind_begin}
        \If{$(G_{i-1}(\Vars,\bar a)\land E(\Vars,\bar a) \land d(\Vars,\bar a)) \land Tr(\Vars, \Vars')\land E(\Vars',\bar a')  \limp d(\Vars',\bar a')$} {
            $c = d$\;\label{alg:genev:rel_ind_end}
        }
    }
}
\Return $c$\;
\end{algorithm}
\end{minipage}
\end{adjustbox}
\hfill
\begin{adjustbox}{scale=0.7}
    \begin{minipage}[t]{7.5cm}
      \begin{algorithm}[H]
\caption{\PushEV}\label{alg:pusher}
\KwIn{$c$ : a clause to propagate, i: int}
\Requires{$c\in G_i$ and $\vG^* = \vG$ where $\vG^*$ holds the initial state of $\vG$ for the post-condition}
\Ensures{$\vG$ is a generalized trace s.t. $\forall j\leq|\vG|\cdot G_j\limp G^*_j$}
\SetKw{Continue}{continue}
\SetKw{Break}{break}
$\mathcal{A} = \emptyset$\;
    \While {True} {
        \eIf{$G_i(\Vars,\bar a)\land E(\Vars,\bar a) \land Tr(\Vars, \Vars')\land E(\Vars',\bar a')  \limp c(\Vars',\bar a')$} {\label{fig:alg:pusher:push}
            $\vG.\textsc{Insert}(i + 1 ,c)$\;
            \Break \;
        } {
            \tcc{extract assignment from \satsolver}
            $\varphi(\bar v', \bar a')$ = \textsc{GetAssignment()}\;
            $\mathcal{A} = \mathcal{A} \cup \{ \varphi \}$\label{fig:alg:pusher:cube}
        }
        $found = false$\;
        \For {$a \in c \text{ where } var(a) \in \bar a$} {
            $C = c.\textsc{Expand}(a)$\;\label{fig:alg:pusher:expand}
            choose $d\in C$ such that $(\Lor\mathcal{A})\not\models \neg d'$\;\label{fig:alg:pusher:choose}
            \If{$d \neq Null$} {
                $c = d$\;\label{fig:alg:pusher:update}
                $found = true$\;
                \Break\;
            }
        

        }
        \lIf {$found = false$} \Break
    }
\end{algorithm}  
    \end{minipage}
\end{adjustbox}
\vspace{-15pt}
\end{figure}

\subsection{Generalization with Auxiliary Variables}

Auxiliary variables are introduced when $\vG$ is re-encoded. Assume that $\bar a$ is a set of auxiliary variables added by re-encoding $\vG$. Since an inductive invariant can now be expressed in terms of $\Vars\cup\bar a$, it is important to allow \PdrEV to generate clauses over $\Vars\cup\bar a$. In \Pdr, new clauses are added to the trace by a generalization procedure. Next, we present \GenEV, which is the generalization procedure used by \PdrEV that allows it to generate clauses over $\Vars\cup \bar a$. 

Recall that \Pdr iteratively tries to prove that a given state\footnote{In practice, it can be a set of states.} is unreachable in $i$ steps, for some $i$. Such a state is referred to as a \emph{proof obligation} at level $i$ and it is represented by a cube and the level as $(\varphi(\Vars), i)$. \Pdr invokes generalization when it identifies that a proof obligation is unreachable. The generalization procedure starts from $\neg\varphi(\Vars)$ as a candidate to be added to the trace at level $i$. 
Then, it tries to find a new clause $c(\Vars)\subseteq \neg\varphi(\Vars)$ by dropping literals from $\neg\varphi$ such that $c$ is inductive relative to frame $i-1$ in the trace. This process is usually referred to as \emph{inductive generalization}~\cite{DBLP:conf/fmcad/BradleyM07,DBLP:conf/vmcai/Bradley11}. The resulting clause $c$ is then added to the trace at level $i$.

\PdrEV does not change the procedure that generates proof obligations. Hence, \GenEV also starts from a proof obligation $(\varphi(\Vars), i)$ that is known to be unreachable. As a first step, \GenEV also uses inductive generalization. Assume that $c(\Vars)\subseteq \neg\varphi(\Vars)$ is the result of inductive generalization. \GenEV tries to generalize $c$ further, by replacing literals that appear in $c$ with literals over auxiliary variables. Recall that a clause that includes auxiliary variables represents a \emph{set of clauses}. Hence, such an operation is a \emph{generalization} as it produces a \emph{stronger} clause.

\begin{example}\label{example:gener}
    Assume $c = (l_1\lor l_2\lor l_3)$ and there exists an auxiliary variable $y$ such that $y \liff l_3\land l_{4}$. Now assume \GenEV replaces $l_3$ with $y$ resulting in $d = (l_1\lor l_2\lor y)$. Since $((y\liff l_3\land l_4)\land d) \limp c$, in the context of \PdrEV, $d$ is a generalization of $c$.
\end{example}

The generalization procedure \GenEV, used by \PdrEV, is presented in Algorithm~\ref{alg:gen}. The input to \GenEV is an unreachable proof obligation $(\varphi(\Vars), i)$. First, inductive generalization is used in order to generate a clause $c\subseteq \neg\varphi$ (line~\ref{fig:alg:gener:ind_gen}). Then, in order to replace literals in $c$ that are over $\Vars$ with literals over $\bar a$, \GenEV iterates over $\bar a$ in order, starting from $a_1$. Recall that for $a_i,a_j\in\bar a$, if $j > i$ then $a_i$ may appear in the COI of $a_j$ but $a_j$ cannot be in the COI of $a_i$ (see Definition~\ref{def:aux_circ}). For each auxiliary variable $a\liff l_1\star l_2$, \GenEV performs the following. If $\star = \land$ then it checks if \emph{only one} of the literals $l_1$ or $l_2$ are in $c$ (lines~\ref{alg:gener:and_begin}-\ref{alg:gener:and_end}). Assume, w.l.o.g., that $l_1\in c$ and $l_2\not\in c$. Then, \GenEV creates a clause $d$, which is similar to $c$ but with $l_1$ substituted by $a$. If $\star = \lxor$ similar checks are performed (lines~\ref{alg:gener:xor_begin}-\ref{alg:gener:xor_end}). For example, if $l_1\in c$ and $l_2\in c$, then $l_1$ and $l_2$ are substituted by $a$. In case such a substitution is performed, \GenEV then checks if the resulting clause $d$ is inductive relative to $G_{i-1}$. If it is, then $c$ is updated to be $d$ (lines~\ref{alg:genev:rel_ind_begin}-\ref{alg:genev:rel_ind_end}), and the loop continues with the next auxiliary variable. In case $d$ is not inductive relative to $G_{i-1}$, generalization through the current auxiliary variable failed and $c$ remains unchanged. Here too the loop continues with the next auxiliary variable.

Note that this process substitutes literals over $\Vars\cup\bar a$ with literals over $\bar a$ in a \emph{bottom-up} fashion. A similar procedure can be performed \emph{top-down}, however, our experiments show that the bottom-up procedure performs better in practice.

\subsection{Fractional Propagation}

One of the key procedures in \Pdr is \emph{propagation}. During propagation \Pdr tries to ``push'' a clause to a higher level in the trace. Since \PdrEV uses a generalized trace $\vG$, a clause in some frame $G_i$ may include auxiliary variables. Let $c(\Vars,\bar a)$ be such a clause. Since $c$ is over $\Vars\cup\bar a$, it represents a \emph{set of clauses} over $\Vars$. Hence, trying to propagate $c$ is akin to propagating a set of clauses simultaneously. If one clause in that set cannot be propagated, $c$ itself cannot be propagated. Intuitively, this reduces the likelihood of such a clause to be propagated. 

\begin{example}\label{example:pusher}
    Consider the clause $c = (y\lor A)$ where $c\in G_i$ and $y$ is an auxiliary variable such that $y\liff l_1\lxor l_2$. In this case $c$ represents the clauses $c_1 = (l_1 \lor l_2 \lor A)$ and $c_2 = (\neg l_1\lor \neg l_2\lor A)$. Assume that

\begin{align}
G_i(\Vars,\bar a)\land E(\Vars,\bar a)\land\Tr(\Vars,\Vars')\land E(\Vars',\bar a') & \not\limp c(\Vars',\bar a')\label{eq:push:c_fail} \\
G_i(\Vars,\bar a)\land E(\Vars,\bar a)\land\Tr(\Vars,\Vars')\land E(\Vars',\bar a') & \not\limp c_1(\Vars',\bar a') \\ 
G_i(\Vars,\bar a)\land E(\Vars,\bar a)\land\Tr(\Vars,\Vars')\land E(\Vars',\bar a') & \limp c_2(\Vars',\bar a')
\end{align}

    The clause $c$ cannot be propagated because $c_1$ cannot be propagated. However, $c_2$ can be propagated to $G_{i+1}$. If the propagation procedure does not treat auxiliary variables specifically, it can miss the fact $c_2$ can be propagated. As a result, \PdrEV is likely to relearn $c_2$ in higher levels during \textsc{MkSafe}, which is less efficient than propagation.
\end{example}


\PushEV is designed to specifically take into account auxiliary variables during propagation. Unlike the propagation procedure in \Pdr, if \PushEV fails to propagate a clause $c(\Vars,\bar a)$, it tries to identify a subset of the set of clauses $c$ represents (a ``fraction'') that can be propagated. In order to understand the intuition behind \PushEV, let us reconsider Example~\ref{example:pusher}. From Equation~\ref{eq:push:c_fail} we can conclude that $G_i(\Vars,\bar a)\land E(\Vars,\bar a)\land\Tr(\Vars,\Vars')\land E(\Vars',\bar a') \land \neg c(\Vars',\bar a')$ is satisfiable. Let us denote by $\varphi(\Vars',\bar a')$ the cube such that $\varphi\models\neg c(\Vars', \bar a')$. Such a cube can be extracted from the satisfying assignment that shows why $c$ cannot be propagated. 
Since $\varphi \models \neg c(\Vars',\bar a')$, it is easy to show that $\varphi \models (\neg c_1(\Vars',\bar a') \lor \neg c_2(\Vars',\bar a'))\land (y'\liff l_1' \lxor l_2')$. By that, $\varphi \models \neg c_1(\Vars',\bar a') \land (y'\liff l_1' \lxor l_2')$ or $\varphi \models \neg c_2(\Vars',\bar a')\land (y'\liff l_1' \lxor l_2')$. \PushEV uses such a cube (i.e. the assignment) to identify elements in the set of clauses $c$ represents, that cannot be propagated. 
For this example, let us assume that $\varphi \models \neg c_1(\Vars',\bar a') \land (y'\liff l_1' \lxor l_2')$ and $\varphi \not\models \neg c_2(\Vars',\bar a') \land (y'\liff l_1' \lxor l_2')$. In this case, \PushEV determines $c_1$ cannot be propagated and explicitly checks if $c_2$ can be propagated.

This exploration performed by \PushEV depends on the number of auxiliary variables a clause depends on. At each propagation attempt that fails \PushEV uses the satisfying assignment returned by the \satsolver to choose an auxiliary variable to expand (i.e. replace it by its definition), and eliminate a subset of clauses that cannot be propagated. Note that it may be possible for an assignment to rule out the entire set of clauses.

The procedure \PushEV, which implements fractional propagation, is described in Algorithm~\ref{alg:pusher}.
\PushEV receives as input a clause $c$ and the level $i$. \PdrEV enters a loop where at each iteration it tries to push $c$ to level $i+1$ (line~\ref{fig:alg:pusher:push}). If it succeeds, then the clause is added to the corresponding frame and the procedure terminates. If it fails, it extracts the assignment from the \satsolver and stores it in $\mathcal{A}$ (line~\ref{fig:alg:pusher:cube}). Then, it iterates over auxiliary literals in $c$. For such a literal $a$, \PushEV retrieves the set of clauses $c$ represents by expanding the definition of $a$ (line~\ref{fig:alg:pusher:expand}). Next, it chooses a clause $d$ in that set that is not ruled out by all previous assignments that were found during failed propagation attempts (line~\ref{fig:alg:pusher:expand}). If such a clause is found, $c$ is updated to $d$ and now represents a ``fraction'', and the algorithm starts a new iteration, trying to propagate the fraction. If no such clause is found for all auxiliary literals in $c$, the procedure terminates, without propagating any fraction of $c$.

We note that some implementation details are omitted for the simplicity of presentation. For example, if \PushEV finds a fraction clause that can be propagated, it makes sure that this clause is not already implied at higher levels. It is also important to note that this procedure does not necessarily identify all fractions that can be propagated.

%% file: 6_experiments.tex
\section{Experimental Evaluation}

\newcommand{\CUDD}{CUDD\xspace}
\newcommand{\CaDiCal}{CaDiCal\xspace}
\newcommand{\SAT}{SAT\xspace}
\newcommand{\BDD}{BDD\xspace}

We implemented both \Pdr and \PdrEV in a new tool from the ground up\footnote{\url{https://github.com/TechnionFV/CAV_2025_artifact}}
For SAT solving, we used \CaDiCal 2.0, a state-of-the-art open-source SAT solver. For \BDD operations, we used \CUDD, a well-established open-source \BDD library.

To isolate the impact of ER within \Pdr, we implemented \PdrEV as an extension of our \Pdr implementation, minimizing differences beyond the core algorithmic changes described in this paper. We used the \Pdr implementation in ABC \cite{abc}, a widely used open-source hardware synthesis and verification tool, as a baseline. In the following, \Pdr and \PdrEV refer to the implementation in our tool, while \Abc refers to ABC's \Pdr implementation.

\subsection{HWMCC}

We conducted a series of experiments comparing the performance of these three implementations on various Hardware Model Checking Competitions benchmarks (HWMCC); HWMCC’19/20/24, which include 317, 324 and 318 instances, respectively.
\footnote{HWMCC’24 officially includes 319 instances but one is malformed}.
All experiments were executed on machines with AMD EPYC 74F3 CPU and 32GB of memory, under a timeout of 3600 seconds.

Figure \ref{fig:table:hwmcc_results} summarizes the results of our experiments, where in all three sets \PdrEV outperforms both \Pdr and \Abc, particularly on SAFE instances. The difference is particularly apparent in HWMCC'24 where \PdrEV not only solves significantly more instances than \Pdr, but also achieves significantly more unique wins than the other solvers, and has a significantly lower average runtime.

\begin{figure}[h!]
\centering
\begin{adjustbox}{scale=0.8}
    
\begin{tabular}{c|c|c|c|c|c|c}
\hline
\hline
Set& Solver & \# Solved & \#UNSAFE & \#SAFE & \# Unique & Avg. Time \\
\hline

\multirow{3}{*}{HWMCC'19}
& \Abc & 212 & 27 & 185 & {\bf 2} & 1274.6  \\
& \Pdr & 233 & {\bf 40} & 193 & 1 & 1104.4  \\
& \PdrEV & {\bf 236} & {\bf 40} & {\bf 196} & {\bf 2} & { \bf 1082.4}  \\
\hline
& VB & 239 & 43 & 196 & & 1026.9 \\
\hline
\hline

\multirow{3}{*}{HWMCC'20}
& \Abc & 221 & 34 & 187 & 3 & 1274.4  \\
& \Pdr & 231 & {\bf 43 } & 188 & 4 & 1144.0  \\
& \PdrEV & {\bf 236} & 41 & {\bf 195} & {\bf 6} & {\bf 1095.4} \\
\hline
& VB & 243 & 46 & 197 & & 1038.4 \\
\hline
\hline

\multirow{3}{*}{HWMCC'24}
& \Abc & 157 & 27 & 130 & 3 & 1920.5  \\
& \Pdr & 175 & 33 & 142 & 0 & 1702.3 \\
& \PdrEV & \bf 188 & \bf 36 & \bf 152 & \bf 9 & \bf 1622.5 \\
\hline
& VB & 191 & 36 & 155 & & 1569.4 \\
\hline
\hline

\end{tabular}
\end{adjustbox}
\caption{\footnotesize HWMCC results, average runtime is in seconds}
\label{fig:table:hwmcc_results}
\vspace{-14pt}
\end{figure}

\paragraph{Runtime:} Figure \ref{fig:hwmcc_all} compares \PdrEV and \Pdr on the entire benchmark set (VB stands for Virtual Best). Runtime comparison for \PdrEV and \Pdr appears in Figure~\ref{fig:scatter_plot_all}. Instances where \PdrEV outperforms \Pdr appear below the parity diagonal, highlighting its advantage. The plot demonstrates that \PdrEV achieves substantial runtime improvements across numerous benchmarks while successfully solving a significant subset of problems that \Pdr fails to solve within the time limit. Some instances exhibit a faster runtime under \Pdr, which we associate with the computational overhead of identifying and using auxiliary variables, particularly those that ultimately prove ineffective for the proof. Overall, \PdrEV performs better than \Pdr as is evident in Figure~\ref{fig:cactus_19_20_24}. Note that \PdrEV solves 550 instances in 1500 seconds while \Pdr takes 2200 seconds to solve the same amount of instances. Moreover, we can see for the harder instances (runtime of over 500 seconds), \PdrEV has a clear advantage over \Pdr. Lastly, \PdrEV is very close to VB. This leads us to conclude that \PdrEV is overall better than \Pdr and in addition its performance does not degrade due to the use of ER. 

\begin{figure}[h!]
\vspace{-5pt}
    \centering
    \begin{subfigure}[t]{0.45\textwidth}
        \centering
        \begin{adjustbox}{scale=0.75}
        \includegraphics[width=0.9\textwidth]{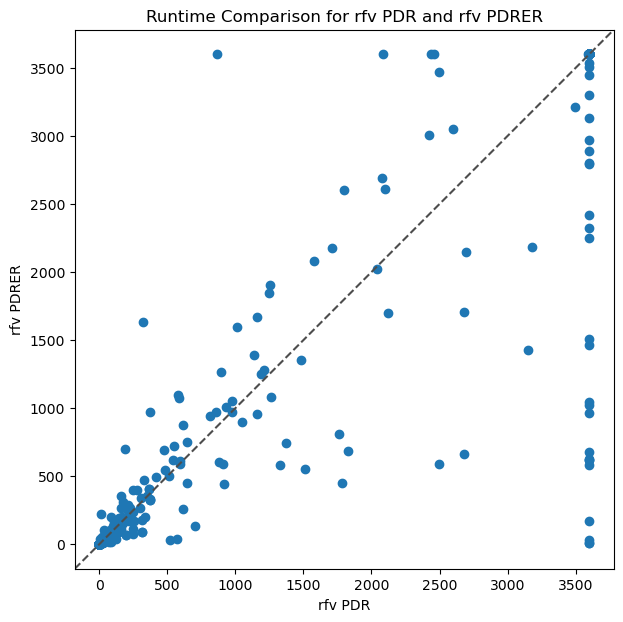}
        \end{adjustbox}
        \caption{Runtime comparison: \Pdr (X-axis) vs. \PdrEV (Y-axis).}
        \label{fig:scatter_plot_all}
    \end{subfigure}
    \begin{subfigure}[t]{0.45\textwidth}
        \centering
        \includegraphics[width=0.88\textwidth]{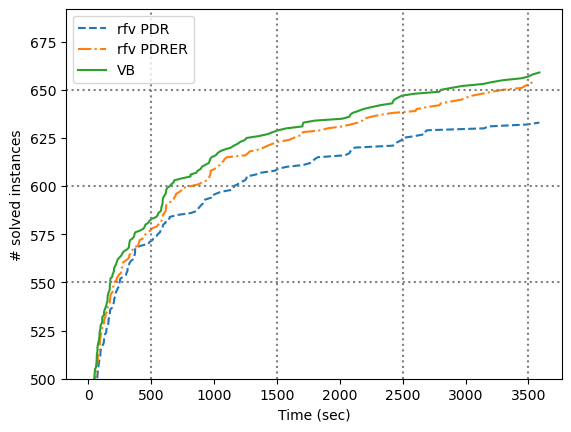}
        \caption{Y-axis represents number of solved instances. VB stands for Virtual Best.}
        \label{fig:cactus_19_20_24}
    \end{subfigure}
    \caption{\footnotesize HWMCC'19/20/24}\label{fig:hwmcc_all}
    \vspace{-18pt}
\end{figure}

\paragraph{Auxiliary Variables: } In order to evaluate the effect ER has on the proofs generated by \PdrEV, we first evaluate the auxiliary variables \PdrEV uses. Note that for trivial instances (i.e., those that are solved quickly), \PdrEV never re-encodes the trace and therefore it does not use auxiliary variables. Figure~\ref{fig:avs} presents the instances that use auxiliary variables and the average number of auxiliary variables used in each instance. Considering SAFE instances, the table presents how many invariants use auxiliary variables and what is the average number of auxiliary variables used to express the invariants. 

\begin{figure}[h!]
\centering
\begin{adjustbox}{scale=0.8}
    
\begin{tabular}{c|c|c|c|c|c|c|c}
\hline
\hline
\multirow{2}{*}{Set} & \multirow{2}{*}{\#Instances} & \#Instances &  Average & \multirow{2}{*}{\%\ $\lxor$} & Invariants & Average AVs & \multirow{2}{*}{\%\ $\lxor$}\\
& & using AVs & AVs &  & using AVs & in Invariants &\\
\hline
HWMCC'19 & 317 & 199 & 59 & 85\% & 94 & 23 & 78\% \\
\hline
HWMCC'20 & 324 & 207 & 56 & 85\% & 105 & 17 & 78\% \\
\hline
HWMCC'24 & 318 & 185 & 76 & 74\% & 61 & 39 & 73\%\\
\hline
\hline

\end{tabular}
\end{adjustbox}
\caption{\footnotesize Auxiliary Variables (AVs) added and used by \PdrEV. \%$\lxor$ stand for percentage of AVs using $\lxor$ in their definition.}
\label{fig:avs}
\vspace{-14pt}
\end{figure}

\paragraph{Proof Size: } In order to evaluate our conjecture regarding the ability of \PdrEV to produce shorter proofs due to the use of ER, we analyzed and compared the proofs generated by both \PdrEV and \Pdr. Figure~\ref{fig:invar_2024} and Figure~\ref{fig:invar_19_20} compare the size of the invariants for HWMCC'24 and HWMCC'19/20, respectively. As can be seen from the table in Figure~\ref{fig:table:hwmcc_results}, \PdrEV outperforms \Pdr considerably on HWMCC'24, while on HWMCC'19/20 the margin in favor of \PdrEV is smaller. A close analysis of the proofs generated by \PdrEV vs. those generated by \Pdr reveals that on HWMCC'24, the proofs generated by \PdrEV are \emph{shorter}. 

Figure~\ref{fig:trace_24} and Figure~\ref{fig:trace_19_20} compare the number of clauses in the trace maintained by \PdrEV ($\vG$) and \Pdr ($\vF$) for HWMCC'24 and HWMCC'19/20, respectively. The plots show a clear advantage for \PdrEV, which generates shorter bounded proofs for the majority of instances. The advantage is more apperant on HWMCC'24, which can explain why \PdrEV is considerably better on that set. This is also in accordance with the data that appears in Figure~\ref{fig:avs}. To further establish the advantage of \PdrEV, Figure~\ref{fig:po_24} and Figure~\ref{fig:po_19_20} compare the number of proof obligations generated by \PdrEV and \Pdr. This metric demonstrates that \PdrEV explores the state-space more efficiently.

\begin{figure}[h!]
    \centering
    \begin{subfigure}{0.32\textwidth}
        \centering
        \includegraphics[width=0.95\textwidth]{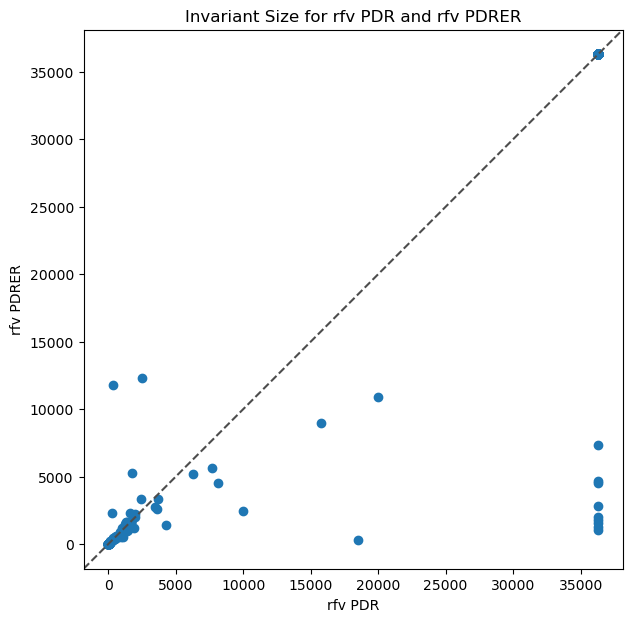}
        \caption{\centering Invariant Size}
        \label{fig:invar_2024}
    \end{subfigure}
    \begin{subfigure}{0.32\textwidth}
        \centering
        \includegraphics[width=1\textwidth]{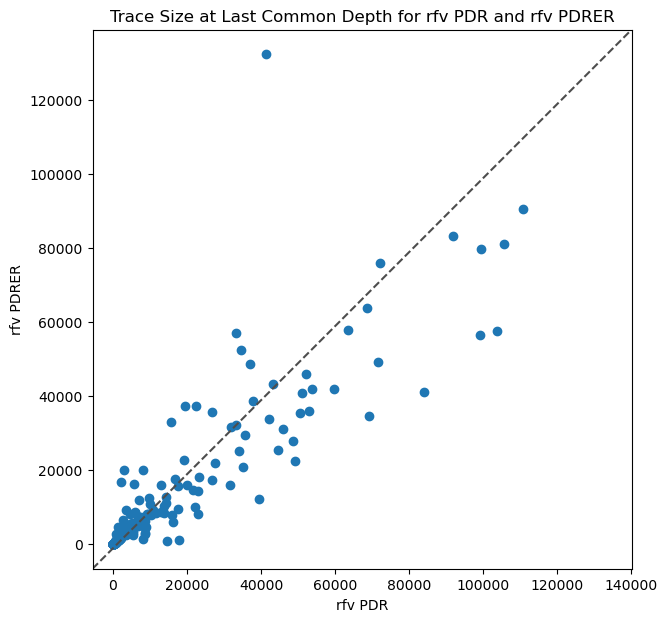}
        \caption{\centering Trace Size}
        \label{fig:trace_24}
    \end{subfigure}
    \begin{subfigure}{0.32\textwidth}
        \centering
        \includegraphics[width=0.93\textwidth]{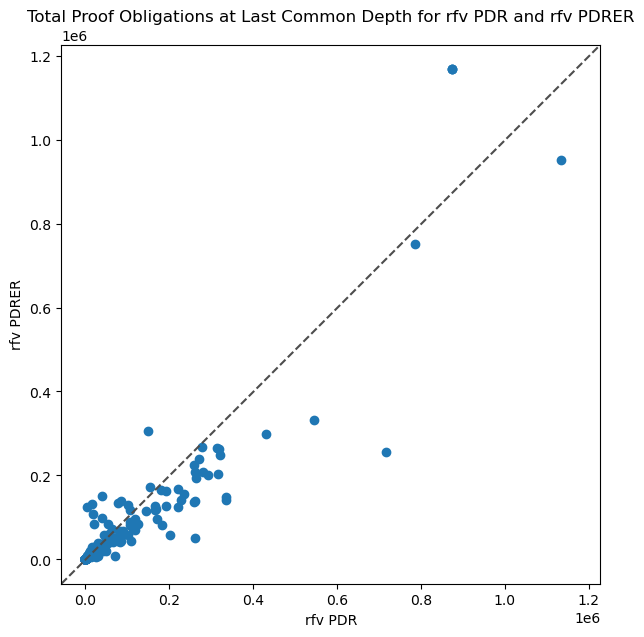}
        \caption{\centering Proof Obligations}
        \label{fig:po_24}
    \end{subfigure}
    \begin{subfigure}{0.32\textwidth}
        \centering
        \includegraphics[width=0.98\textwidth]{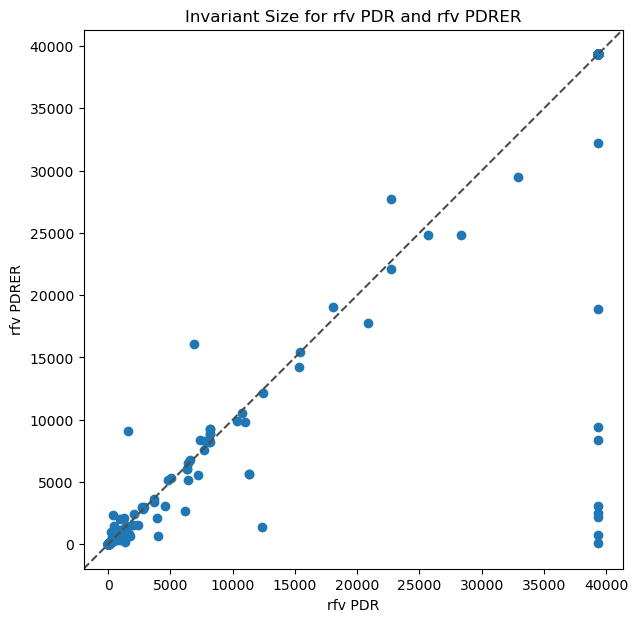}
        \caption{\centering Invariant Size}
        \label{fig:invar_19_20}
    \end{subfigure}
    \begin{subfigure}{0.32\textwidth}
        \centering
        \includegraphics[width=1\textwidth]{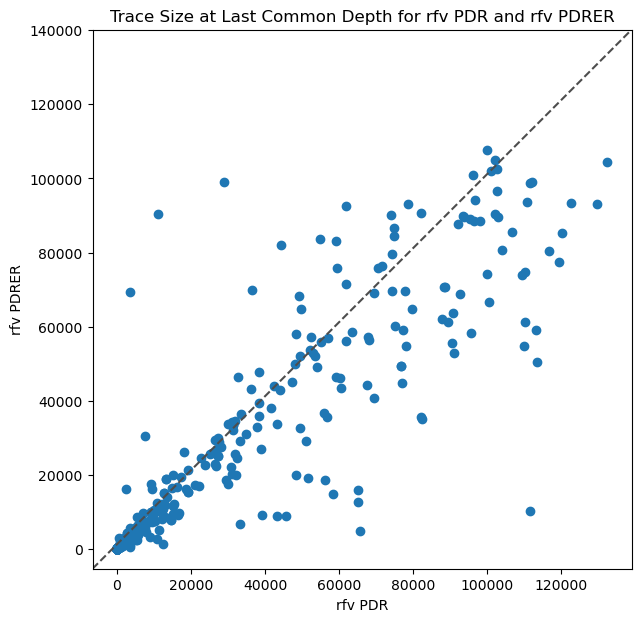}
        \caption{\centering Trace Size}
        \label{fig:trace_19_20}
    \end{subfigure}
    \begin{subfigure}{0.32\textwidth}
        \centering
        \includegraphics[width=0.95\textwidth]{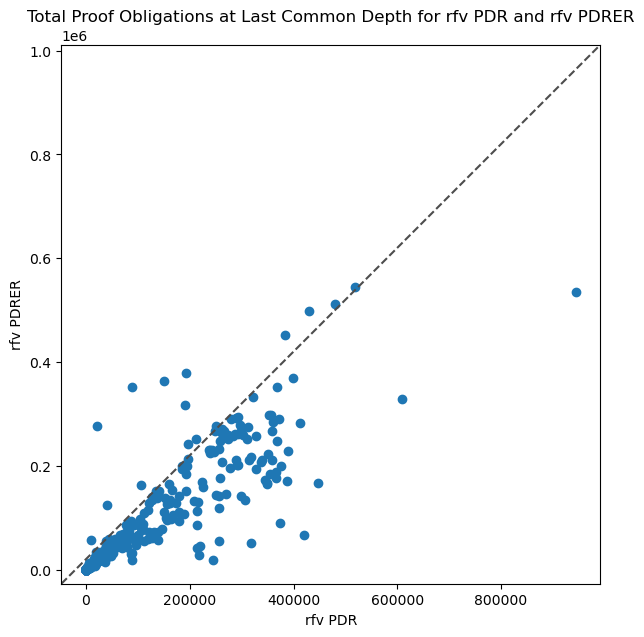}
        \caption{\centering Proof Obligations}
        \label{fig:po_19_20}
    \end{subfigure}
    \caption{\footnotesize Proof comparison: \PdrEV vs. \Pdr. X-axis is \Pdr across all plots. (\ref{fig:invar_2024})-(\ref{fig:po_24}) HWMCC'24, (\ref{fig:invar_19_20})-(\ref{fig:po_19_20}) HWMCC'19/20.}
    \vspace{-15pt}
\end{figure}


\subsection{The Effect of \ImpEV, \GenEV and \PushEV}

In order to understand the importance of the adaptions made to the different parts of \Pdr, we analyze the impact these modifications have on performance.

Table~\ref{tbl:configs} summarizes the contribution of each of these functions to \PdrEV on the entire benchamrk set (HWMCC'19/20/24). The letters ``G'', ``F'', and ``I'' stands for \GenEV, \PushEV and \ImpEV, respectively. The second row in the table (\PdrEV without ``G'', ``F'', and ``I'') represents \PdrEV without any modifications to core \Pdr functions. This version of \PdrEV only uses ER to re-encode the trace, and uses the standard \Pdr generalization, propagation and subsumption checks. As can be seen from the table, it solves the least number of instances and has the highest runtime. Adding \GenEV (\PdrEV-F-I) improves the number of solved instances and the runtime. Similarly, adding \PushEV and \ImpEV increases the number of solved instances. Interestingly, considering the entire benchmark, \PdrEV-F performs the best in terms of runtime. We believe that this is only due to the implementation and can be improved.

\begin{table}[h]
\centering
\begin{adjustbox}{scale=0.8}
\begin{tabular}{c|c|c}
\hline
\hline
Configuration & Number of Solved Instances & Average Runtime [s] \\
\hline
\Pdr & 639 & 1316 \\
\hline
\hline
\PdrEV-G-F-I & 652 & 1297 \\
\hline
\PdrEV-F-I & 655 & 1264 \\
\hline
\PdrEV-I & 657 & 1273 \\
\hline
\PdrEV-F & {\bf 660} & {\bf 1234} \\
\hline
\PdrEV & {\bf 660}  & 1283 \\
\hline
\hline
\end{tabular}
\end{adjustbox}
\caption{Performance with and without \PdrEV modifications.}
\label{tbl:configs}
\vspace{-14pt}
\end{table}

\paragraph{BDDs and \ImpEV:} As mentioned in Section~\ref{sec:pdrer}, syntactic subsumption checks are insufficient when extension variables are used. On this benchmark set, syntactic checks in \ImpEV are only sufficient in 40\% of the cases ($true$ is returned due to line~\ref{alg:implication:subsumption}). Other syntactic checks (lines~\ref{alg:implication:subsumption}-\ref{alg:implication:coi_end}) are successful in 94\% of the times, and BDDs are only used in 6\% of the calls to \ImpEV.

\subsection{Exponential Invariant in CNF}

The closest work to ours is~\cite{DBLP:conf/fmcad/DurejaGIV21}, where \Pdr is extended such that inductive invariants can be expressed in terms of $\Vars$ as well as logical connectives appearing in $\Tr$ (i.e., internal signals). Since the tool (\textsc{IC3-INN}) developed in~\cite{DBLP:conf/fmcad/DurejaGIV21} is not open-source and not available, we could not compare \PdrEV against it. However, we did evaluate \PdrEV on the examples described in Section III of \cite{DBLP:conf/fmcad/DurejaGIV21}. The examples admit only an exponential invariant in CNF and are therefore challenging for \Pdr, but can be solved by \textsc{IC3-INN}. We confirm that \PdrEV can solve these examples by finding compact inductive invariants, as expected. 

Furthermore, 
we identified instances where the inductive invariant cannot be expressed compactly by $\Vars$ and internal signals. For such examples, \textsc{IC3-INN} does not provide a benefit over \Pdr. Let us consider \texttt{vis\_arrays\_bufferAlloc} from HWMCC'19. It implements a buffer allocation protocol that works as follows. Assume that the buffer has $k$ cells, a user can either request to allocate a cell on the buffer, or free a specific cell, not necessarily in order. Whenever a cell is allocated or de-allocated, a counter is being incremented or decremented by 1, respectively. In addition, a bit-array, called \texttt{busy}, tracks which cells are used and which are free. A simple safety property for such a protocol states that the counter is always in the range $[0,k]$. While this protocol is fairly simple, \Pdr struggles with it since the invariant requires $O(2^k)$ clauses. For example, for $k=16$, \Pdr requires almost an hour to find the invariant, and the invariant contains more than 65,000 clauses. In contrast, \PdrEV solves this instance in 10 minutes, finding an invariant with around 5000 clauses. We emphasize that while we could not evaluate this example with \textsc{IC3-INN}, the invariant cannot be expressed in terms of internal signals, and hence we conjecture that \textsc{IC3-INN} should perform similarly to \Pdr.

%% file: 9_conclusion.tex
\section{Conclusion}

We presented \PdrEV, the \emph{first} model checking algorithm that efficiently uses Extended Resolution as its underlying proof-system.
\PdrEV generalizes \Pdr and includes many algorithmic enhancements that enable an efficient integration of ER in model checking. Due to the use of ER, \PdrEV produces shorter proofs for the majority of instances from HWMCC'19/20/24. In addition, and most importantly, the use of ER in \PdrEV does not lead to performance degradation in the general case and outperforms \Pdr in most instances we evaluated. Moreover, it admits short proofs for problems for which \Pdr can only admit proofs of exponential size. We strongly believe that \PdrEV demonstrates that strong proof systems can be used efficiently in model checking.
